\documentclass[prl,twocolumn,superscriptaddress,nofootinbib,longbibliography]{revtex4-2}
\usepackage{custom_style}

\newif\ifptitle
\newif\ifpnumber
\newcounter{para}

\ptitletrue  
\pnumbertrue  

\begin{document}


\title{Leveraging Symmetry Merging in Pauli Propagation}
\date{\today}

\author{Yanting Teng}
\affiliation{Institute of Physics, Ecole Polytechnique F\'{e}d\'{e}rale de Lausanne (EPFL),   Lausanne, Switzerland}
\affiliation{Centre for Quantum Science and Engineering, Ecole Polytechnique F\'{e}d\'{e}rale de Lausanne (EPFL), Lausanne, Switzerland}
\email{yanting.teng@epfl.ch}

\author{Su Yeon Chang}
\affiliation{Institute of Physics, Ecole Polytechnique F\'{e}d\'{e}rale de Lausanne (EPFL),   Lausanne, Switzerland}
\affiliation{Theoretical Division, Los Alamos National Laboratory  (LANL), Los Alamos, United States}

\author{Manuel S. Rudolph}
\affiliation{Institute of Physics, Ecole Polytechnique F\'{e}d\'{e}rale de Lausanne (EPFL),   Lausanne, Switzerland}
\affiliation{Centre for Quantum Science and Engineering, Ecole Polytechnique F\'{e}d\'{e}rale de Lausanne (EPFL),   Lausanne, Switzerland}

\author{Zo\"{e} Holmes}
\affiliation{Institute of Physics, Ecole Polytechnique F\'{e}d\'{e}rale de Lausanne (EPFL),   Lausanne, Switzerland}
\affiliation{Centre for Quantum Science and Engineering, Ecole Polytechnique F\'{e}d\'{e}rale de Lausanne (EPFL),   Lausanne, Switzerland}
\affiliation{Algorithmiq Ltd, Kanavakatu 3 C, FI-00160 Helsinki, Finland}

\begin{abstract}
We introduce a symmetry-adapted framework for simulating quantum dynamics based on Pauli propagation. When a quantum circuit possesses a symmetry, many Pauli strings evolve redundantly under actions of the symmetry group. We exploit this by merging Pauli strings related through symmetry transformations. This procedure, formalized as the \textit{symmetry-merging Pauli propagation} algorithm, propagates only a minimal set of orbit representatives. Analytically, we show that symmetry merging reduces space complexity by a factor set by orbit sizes, with explicit gains for translation and permutation symmetries. Numerical benchmarks of all-to-all Heisenberg dynamics confirm improved stability, particularly under truncation and noise. Our results establish a group-theoretic framework for enhancing Pauli propagation, supported by open-source code demonstrating its practical relevance for classical quantum-dynamics simulations.
\end{abstract}

\maketitle

\vspace{1mm}

\section{Introduction}
Understanding the dynamics of quantum operators lies at the heart of quantum many-body physics and quantum information theory. Among many practical applications, operator dynamics provide physical insights into quantum information scrambling, many-body chaos, and dynamical correlations in interacting systems for quantum metrology and sensing~\cite{plodzien2025lecture, sieberer2019digital,  nandy2025quantum, montenegro2025quantum}. 

There now exists a broad zoo of classical simulation methods for approximating quantum dynamics. Matrix product states (MPS) have been remarkably successful in one dimension~\cite{Paeckel_2019, Haegeman_2011, Vidal_2004, Verstraete_2004, White_2004}, and a variety of higher-dimensional generalizations are advancing rapidly~\cite{mandra2025heuristic,Alkabetz_2021, Tindall_2023,tindall2025dynamics,rudolph2025simulating,gray2025tensor}. Additional theoretical frameworks include neural quantum states~\cite{Carleo_2017, carleo2024simulatingadiabaticquantumcomputation}, mean-field approaches to operator hydrodynamics~\cite{khemaniOperatorSpreadingEmergence2018, rakovszkyDiffusiveHydrodynamicsOutTimeOrdered2018}, and specialized semiclassical techniques such as the discrete truncated Wigner approximation~\cite{wurtz2018cluster}. Yet no single “winner-takes-all” method exists that reliably handles all problems.

\begin{figure}[t!]
    \centering
    \includegraphics[width=\linewidth]{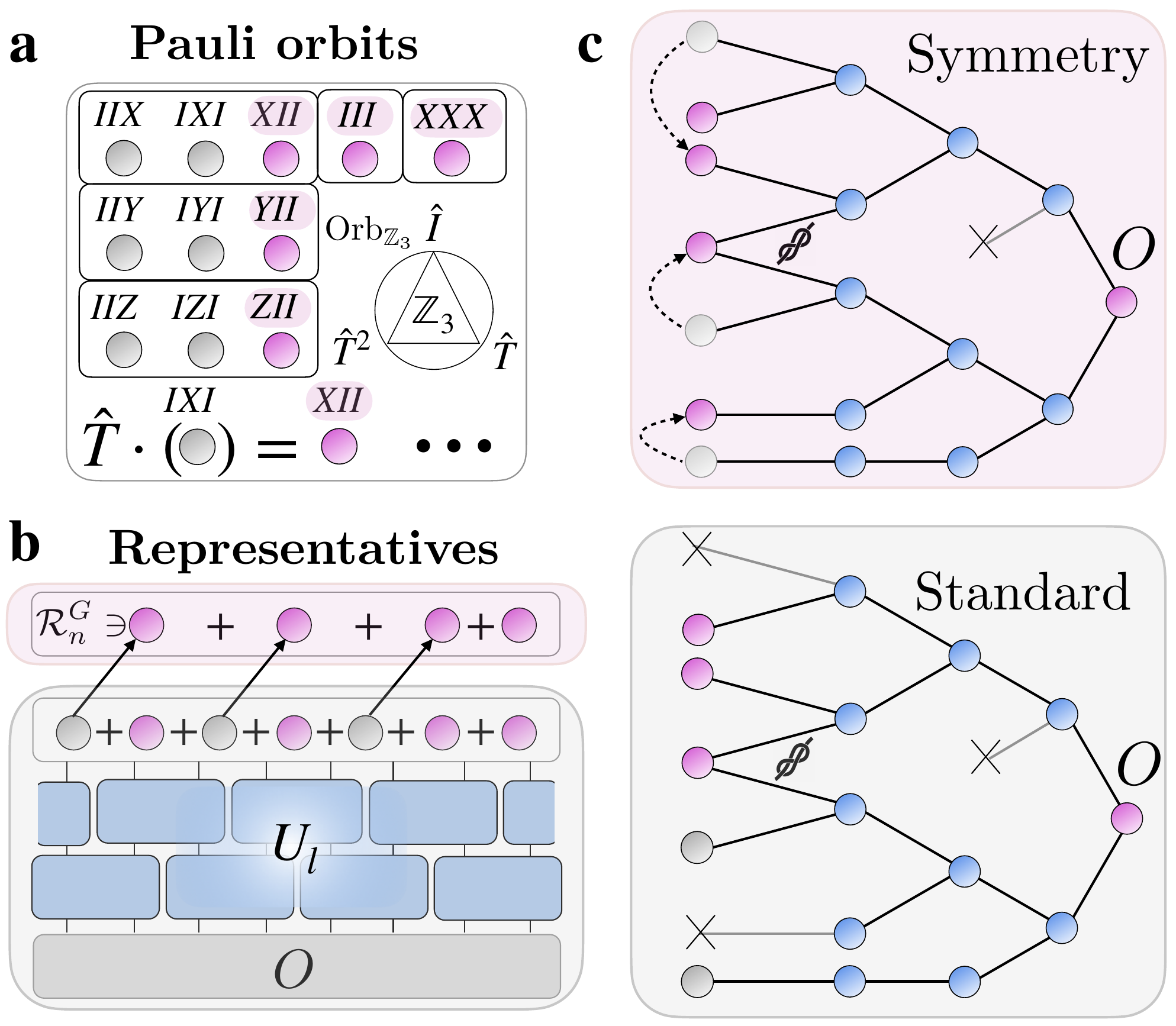}
    \caption{\textit{Schematics of symmetry-merging Pauli propagation.} \textbf{a}) The symmetry group (e.g., $\mathbb{Z}_3$ generated by $\hat{T}$) partitions the set of all $n$-qubit Pauli strings into orbits ($\operatorname{Orb}_{\mathbb{Z}_3}$). Each orbit can be represented by a single \textit{orbit representative} (pink circles).  \textbf{b}) Instead of operating on the full space, symmetry PP evolves a minimal set of representatives $\mathcal{R}^G_n$ as in \equref{eq:O-merging}. After $l$-th unitary layer, the evolved operator $U_l^\dagger O U_l$ contains many Paulis. The merging of symmetry-equivalent Paulis by their orbit representatives reduces the number of distinct terms. \textbf{c}) Comparison of the symmetry PP (top) and standard Pauli propagation (bottom). After each symmetric layer, non-representative Paulis (gray) are merged to orbit representatives (pink), potentially recovering terms that would otherwise be truncated (crossed nodes). 
    }
    \label{fig:fig1}
\end{figure}

In many scenarios, it is effective to study the evolution of observables in the Heisenberg picture rather than tracking the full state evolution.  
This idea has been explored using MPS to simulate operators dynamics in one dimension via dissipation-assisted operator evolution (DAOE)~\cite{rakovszkyDissipationassistedOperatorEvolution2022}. Another approach to approximate operator dynamics is through \emph{Pauli propagation} (PP), which tracks the coefficients of Pauli operators under successive circuit layers~\cite{rallSimulationQubitQuantum2019,aharonov2022polynomial, beguvsic2023simulating, fontana2023classical, shao2023simulating, rudolph2023classical, schuster2024polynomial, angrisani2024classically, gonzalez2024pauli, lerch2024efficient, cirstoiu2024fourier, angrisani2025simulating, fuller2025improved, rudolph2025simulating, angrisani2025simulating, rakovszkyDissipationassistedOperatorEvolution2022}. Pauli propagation encodes Pauli strings as bitstrings and provides a fast bitwise processing of evolution for dynamical quantities involving two procedures: 1) truncating irrelevant Pauli operators and 2) merging Paulis that are identical~\cite{rudolph2025pauli}. However, as the system evolves, the number of nonzero Pauli coefficients typically grows exponentially, leading to severe memory and runtime bottlenecks that limit standard Pauli propagation~\cite{loizeauQuantumManybodySimulations2024, dowling2025bridging, dowling2025magic}. 


In this work, we present a symmetry-adapted framework for Pauli propagation by modifying the merging subroutines based on group theory as sketched in Fig.~\ref{fig:fig1}. In particular, we focus on discrete symmetries in many physically relevant quantum circuits and Hamiltonians such as translation symmetry. By merging Pauli operators connected by symmetry, we can substantially reduce the number of distinct coefficients that must be tracked during evolution. In \thmref{thm:symmetry}, we formalize the invariance of expectation values under such transformations, ensuring that merging does not alter physical predictions. We further quantify the resulting reduction in the number of propagated Pauli terms in \thmref{thm:space-complexity}, showing that the computational space requirement improves by a factor proportional to the size of the symmetry group in the asymptotic limit of large system size.

Our \textit{symmetry-merging Pauli propagation} (symmetry PP) algorithm generalizes naturally across any discrete symmetry groups and can be combined with noise models or other truncation schemes, as illustrated in our numerical examples. The resulting method achieves both improved stability and reduced space complexity,  while memory is a key limiting factor of the \textit{standard Pauli propagation} (standard PP). Various works have exploited symmetries in one way or another~\cite{loizeauQuantumManybodySimulations2024,loizeau2025opening,d2025circuit,haghshenas2025digital, parker2019universal}. We present the first principled Pauli propagation algorithm with exhaustive symmetry merging throughout the simulation and with principled guarantees for correctness and performance. 

\medskip
\paragraph*{The Role of Symmetry.} Consider a unitary $U$ composed of $L$ layers $U = \prod_{l=1}^{L} U_l$. For instance, the $l$-th unitary layer $U_l$ can be a trotterized Hamiltonian layer $U_l = \prod_{j} e^{-i h^{(1)}_j t}e^{-i h^{(2)}_j t}$ for the translational invariant Hamiltonian $H = \sum_j h^{(1)}_j + h^{(2)}_j$ or only a sublayer $U_l = \prod_{j} e^{-i h^{(1)}_j t}$. Let $G$ be a finite symmetry group and $A_g$ be the representation of $g \in G$. Each layer is invariant under $G$ defined by the \emph{group action} $g \cdot(\,)$:
\begin{equation}\label{eq:action}
    g \cdot (U_l) = U_l \Longleftrightarrow A_g U_l A_g^{-1} = {U_l}, \quad \forall \, g \in G.
\end{equation}

The main goal of this paper is to more efficiently determine how operators evolve under these symmetric unitaries. 
If the initial state ($\rho$) is also symmetric such that $g \cdot (\rho) = \rho$, then expectation values of normalized Pauli strings $s \in \mathcal{P}_n:=1/\sqrt{2^n}\{I, X, Y, Z\}^{\otimes n}$ related by a symmetry transformation are identical by the cyclic property of trace,
\begin{equation}\label{eq}
\Tr[\rho \,s] = \Tr[g \cdot (\rho) s] = \Tr[\rho \,g \cdot (s)], \quad \forall g \in G.
\end{equation}
Intuitively, evolving a Pauli string $g\cdot(s)$ related by symmetry is the same as evolving $s$.

Under the evolution of a symmetric unitary, this redundancy allows observables to be computed from a reduced set of representative $n$-qubit Pauli strings, denoted as $\mathcal{R}^G_n$. While we focus on unitary dynamics, our framework generalizes for open quantum systems. We illustrate this flexibility in our numerical example by adding noise layers modeling a local Lindbladian dynamics~\cite{lidar2019lecture, breuer2002theory}.

\section{Symmetry-merging Pauli propagation}
Given an initial state $\rho$ and the unitary $U = \prod_{l=1}^L U_l$ with discrete symmetry $G$, our goal is to compute the expectation value of the evolved target observable $\Tr[\rho\, U^\dagger O U]$. 
At a high level, the \emph{standard PP} algorithm evolves a target observable $O$ using the Heisenberg equation of motion. The evolved operator can be expanded in the Pauli basis as a \textit{Pauli path integral}, 
\vspace{-1mm}
\begin{equation}
     U^{\dagger} O U = \sum_{s_0, \cdots s_L \in \mathcal{P}_n} \Tr[s_L O] \prod_{l=1}^{L} \Tr[s_{l-1} U_l^{\dagger} s_{l} U_l] s_0.
\end{equation}
Here, the summation over normalized Pauli strings $s_l \in \mathcal{P}_n$ enumerates all possible paths during the evolution (see \appref{app:pp}). 
Here we propose avoiding redundant computations by tracking only a minimal set of representative Pauli strings from each orbit under $G$ by introducing \textit{symmetry merging}.

\vspace{1mm}

\noindent\rule{8.5cm}{0.8pt}
\vspace{-2mm}
\noindent
\textbf{Algorithm 1}  

\vspace{1mm}
\noindent Symmetry-merging Pauli propagation (symmetry PP)
\noindent\rule{8.5cm}{0.8pt}
\label{alg:symmetry}
\begin{enumerate} 
	\item Prior to evolution, merge Paulis that are related by a symmetry transformation such that 
    \begin{equation}
        \widetilde{O}_L := \sum_{s \in \mathcal{R}^{G}_n} \sum_{g \in G}\Tr[O\, g \cdot(s)] s \, .
    \end{equation}
	\item Symmetry-merging propagation: For each layer $l=L, L-1, \cdots, 1$, to the following:
    \begin{enumerate}
        \item Compute the Heisenberg evolved observable $O_{l-1} = U_l^{\dagger} \widetilde{O}_{l} U_l$ in the Pauli basis. 
        \item Merge Paulis by orbit representatives $\mathcal{R}^{G}_n$:
     \begin{equation}\label{eq:O-merging}
         \widetilde{O}_{l-1} := \sum_{s \in \mathcal{R}^{G}_n}\sum_{g \in G} \Tr[O_{l-1}\, g \cdot (s)]s.
     \end{equation}
	\label{alg:merging} 
    \end{enumerate}
    \vspace{-5mm}
	\item Compute the inner product between the initial state and the merged observable $\Tr[\rho \, \widetilde{O}_0]$. 
\end{enumerate}
\vspace{-4mm}
\noindent\rule{8.5cm}{0.8pt}

More precisely, consider any Pauli strings $s \in \mathcal{P}_n$ throughout the evolution, we use $\operatorname{Orb}_{G}(s):= \{ g\cdot(s)\, |\, g \in G\}$ to denote the \emph{Pauli orbits} of $s$ under the group action defined in \equref{eq:action}. The \emph{Pauli orbit representatives}, denoted as $\mathcal{R}^G_{n} \subset \mathcal{P}_{n}$, provide a natural minimal set of Pauli strings that need to keep track of. The orbits of a full set of orbit representatives form a disjoint union of all Pauli strings $\mathcal{P}_n = \coprod_{s\in \mathcal{R}^G_{n}} \operatorname{Orb}_{G}(s)$.

For example, as sketched in Fig.~\ref{fig:fig1}, we consider a $3$-qubit system under symmetry group $\mathbb{Z}_3$, represented by translation operators $\hat{I}, \hat{T}, \hat{T}^2$. Here, the group action is defined by translating a Pauli string to the left, e.g. $\hat{T}\cdot(IXI) = XII$. The Pauli orbit of an (unnormalized) string $XII$ is $\operatorname{Orb}_{\mathbb{Z}_3}(XII) = \{IIX, IXI, XII\}$ as illustrated in \figref{fig:fig1}\textbf{a}. The set of complete orbit representatives $\mathcal{R}^{G}_n = 1/\sqrt{2^n} \{III, XII, YII, ZII, \dots, XXX, \cdots \}$ contains $24$ elements and is given by \exref{ex:Pauli-orbits-zn}. For more detailed definitions and a review of group theory, we refer the reader to \appref{app:group}. 

At a high level, symmetry PP involves two subroutines: a) perform standard PP; b) merge Paulis to their orbit representatives $\mathcal{R}^G_{n}$ (\figref{fig:fig1}\textbf{b} and \figref{fig:fig1}\textbf{c}). It is formally given by Algorithm~\hyperref[alg:symmetry]{1}.
In particular, the merging subroutine of step~\hyperref[alg:merging]{2} identifies all Pauli strings equivalent under symmetry transformations and merges them into their corresponding representative of their orbit. The key intuition is that when a quantum circuit exhibits symmetry, many Pauli strings contribute to the expectation value equally under the unitary transformation.

In \appref{app:merging-subroutines}, we outline merging subroutines for step~\ref{alg:merging} and discuss their trade-offs compared to standard PP. More specifically in \appref{app:merging_subroutines}, we discuss the numerical implementations of the merging subroutine and how they can be practically more efficient than the most generic merging subroutine.

\section{Analytical guarantees}
\noindent 
We provide rigorous justifications to symmetry PP algorithm outlined in the previous section. First, we show that symmetry PP outputs the correct result for the dynamics of an operator in \thmref{thm:symmetry}. Then we compare its space complexity comparing to the standard PP in \thmref{thm:space-complexity}. Additionally, we comment on time complexity in \appref{app:complexity}. 

\begin{theorem}[\textbf{Symmetry PP}]\label{thm:symmetry}
    Let $\rho$ be a symmetric initial state and $U = \prod_{l=1}^L U_l$ be the $L$-layer unitary. Suppose each $l$-th layer has a discrete symmetry $G$ with representations $A_g$ such that $\commu{A_g}{\rho} = \commu{A_g}{U_l}=0, \forall g \in G$. Let $\Tr[\rho\, U^\dagger O U]$ be the target expectation value of an observable $O$. Then the expectation value of the merged operator given by Algorithm~\hyperref[alg:symmetry]{1} as $\widetilde{O}_0$ equals the target value
    \begin{equation}
        \Tr[\rho \, \widetilde{O}_0] = \Tr[\rho\, U^\dagger O U].
    \end{equation}
\end{theorem}
\noindent 
Computing the expectation value using only the Pauli representatives is sufficient under symmetric evolutions. 
We provide a proof of \thmref{thm:symmetry} in \appref{app:symmetry-merging}. While the present theorem is applicable to exact evolution, we perform numerical experiments and analyze the effects of truncations commonly adopted in standard PP. A major limiting factor of standard PP is the memory computational cost of storing Pauli strings being propagated. We now discuss the key advantage of memory saving using  symmetry PP. 

\begin{theorem}[\textbf{Space complexity} (informal)]\label{thm:space-complexity}
        Let $G$ be the discrete symmetry group of an $n$-qubit system. The symmetry PP defined by Algorithm~\hyperref[alg:symmetry]{1} reduces the space complexity lower bounded by a ratio 
        \begin{equation}
            r = \frac{\abs{\mathcal{R}^{G}_n}}{4^n} < 1,
        \end{equation} 
        compared to standard PP, where $\abs{\mathcal{R}^{G}_n}$ corresponds to the number of Pauli orbit representatives. 
\end{theorem}

\begin{figure}[t!]
    \centering
    \includegraphics[width=0.9\linewidth]{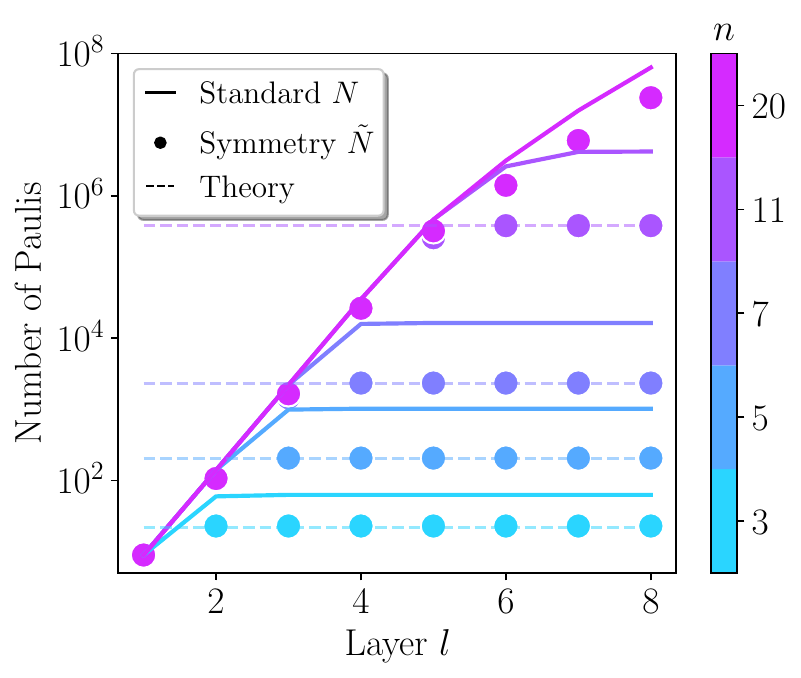}
    \caption{\textit{Reduced Paulis in $1d$ Ising dynamics.} Scaling of the number of propagated Pauli strings with circuit layers $l$ for system sizes from $n = 3$ to $n=20$ (color bar). Lines plot standard PP $N$ (Standard) and circles denote the symmetry PP $\tilde{N}$ (Symmetry) under translation symmetry. Dashed lines show theoretical predictions for the number of distinct Pauli terms. Symmetry merging substantially suppresses the exponential growth of propagated Paulis, saturating to the theoretical scaling of $\frac{1}{n}(4^n + 4(n-1)))$ terms discussed in Example~\ref{ex:Zn}.}
    \label{fig1:symmetry_Paulis}
\end{figure}

\noindent A smaller ratio represents more saving on Pauli strings. The lower bound $r$ is saturated when the evolved operator spans the full operator space, typically relevant for a long-time evolution under ergodic quantum dynamics.
The number of orbit representatives $\abs{\mathcal{R}^{G}_n}$ depends on the symmetry group $G$ governing the system. While it is difficult to generalize the bound as a function of the number of qubits $n$, we have found that in the asymptotic limit of $n \gg 1$, the space saving ratio is inversely proportional to the group size $r\in \Theta(\frac{1}{\abs{G}})$. In the following, we give two examples and provide their proofs with a summary of space complexity for different groups in \appref{app:complexity} for interested readers. 

\begin{example}[Translation symmetry $\mathbb{Z}_n$ in $1d$]\label{ex:Zn}
 Consider the translational symmetry group $\mathbb{Z}_n$ on $n$ qubits. If $n$ is a prime number, 
\begin{equation}
  r = \frac{1}{n}\left(\frac{4^n + 4(n-1))}{4^n}\right).
\end{equation}
For $n \gg 1$, we have a linear space saving $r \in \Theta(\frac{1}{n})$. 
\end{example}
        
\begin{example}[Permutation symmetry $S_n$]
 Consider the permutation group $S_n$. The space saving ratio is
\begin{equation}
 r = \frac{1}{4^n} \binom{n+3}{3}\;. 
\end{equation} 
For $n \gg 1$, the number of orbit representatives scales as $\Theta(n^3)$, leading to an exponential improvement in space complexity $r \in \Theta\left(\frac{n^3}{4^n}\right)$.
\end{example}

In the next section, we demonstrate the advantage of symmetry PP numerically.

\begin{figure}[t!]
    \centering
    \includegraphics[width=0.9\linewidth]{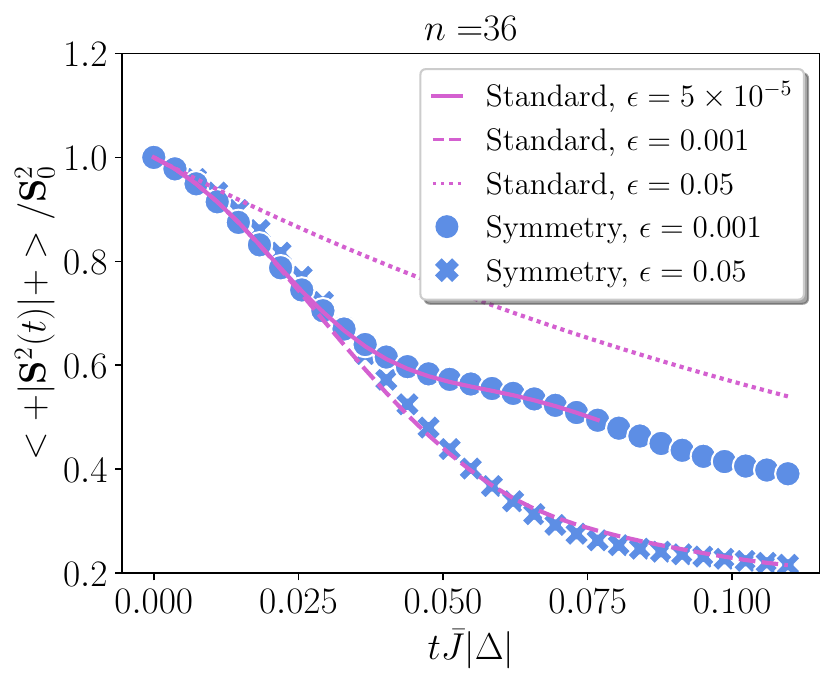}
    \caption{\textit{Spin expectation values of all-to-all Heisenberg dynamics.} We compare predictions from standard PP (lines) and symmetry PP (markers). The total spin operator is defined in \equref{eq:s2} and the normalization factor $\mathbf{S}_0^2 = \expval{\mathbf{S}^2}{\boldsymbol{+}}$ is computed with respect to the initial product state along $x$-axis $\ket{\boldsymbol{+}} := \ket{+}^{\otimes n}$ for $n=36$ qubits with $J_{\perp}=1$ and $\Delta=-1.8$. We also introduce depolarizing noise layers, damping the coefficient by non-identity Pauli weight $e^{-\gamma \abs{P}}$ (e.g. $\abs{XXIZ}=3$). The numerical experiments were performed using up to $100 \text{G}$ memory and up to three days on a single CPU-- a limit that only standard PP with the lowest cutoff $\epsilon$ reached. To achieve the most accurate result, symmetry PP \textit{uses $10$ times fewer Paulis} than standard PP: $\tilde{N}\approx 2 \times 10^8$ Paulis at the end of blue circles for symmetry compared to $N\approx 2 \times 10^7$ Paulis at the end of the pink line for standard PP.}
    \label{fig2:xxz_translation}
\end{figure}

\section{Numerical examples}
\textit{Example: periodic 1$d$ Ising model.} 
The tilted-field Ising model is described by the Hamiltonian:
\begin{equation}
    H_{\text{Ising}} = - \sum_{\langle i j \rangle } Z_i Z_j - h_z \sum_{i=1}^{n} Z_i - h_x \sum_{i=1}^{n} X_i,
\end{equation}
where $\langle i j \rangle$ enumerates the nearest-neighbor sites and $h_x (h_z)$ is the transverse (longitudinal) field. It is known that in the absence of the longitudinal field, the Hamiltonian is integrable using Jordan Wigner transformations~\cite{lieb1961two} such that only a subspace is explored, which is spanned by a polynomial in $n$ number of Pauli operators. Away from the integrable limit ($h_z=0$), this model has been widely used to study thermalization and ergodic dynamics~\cite{hertz2018quantum, roberts2015localized}. In our numerical results, we consider the field strengths $g_x=1.4, g_z=0.9045$ studied in prior works~\cite{rakovszkyDissipationassistedOperatorEvolution2022}. Under periodic boundary conditions, the system has the translation symmetry $\mathbb{Z}_n$.  

The Hamiltonian dynamics is unitary and can be trotterized to $L = \frac{t}{\delta}$ layers as $U = \prod_{l=1}^{L} \prod_{\langle i j \rangle } e^{-i \delta Z_i Z_j} \prod_{i} e^{- i h_z \delta  Z_i} \prod_{i} e^{- i h_x \delta X_i}$. We only merge after a full layer of Hamiltonian trotterization, hence the merging procedure can be further improved if we merge after a minimal symmetric layer, e.g. a single layer of $ZZ$, or $Z$, or $X$ Pauli gates. Consider an initial Pauli operator $Z_{\lceil n/2 \rceil}$ located in the middle of the periodic chain at site $\lceil n/2 \rceil$, the evolution of this operator is given by $U^\dagger Z_{\lceil n/2 \rceil} U = \sum_{\alpha=1}^{N} c_{\alpha} s_{\alpha}$, where $\alpha$ enumerates Pauli strings $s_{\alpha}$ with coefficients $c_{\alpha}$ and $N$ denotes the number of Pauli strings with nonzero coefficients. Using symmetry PP, the merged operator $\widetilde{Z}_{\lceil n/2 \rceil} = \sum_{\alpha=1}^{\tilde{N}} \tilde{c}_{\alpha} \tilde{s}_{\alpha}$ has a reduced number $\tilde{N}$ of Pauli representatives $\tilde{s}_{\alpha} \in \mathcal{R}^G_n$ by merging contributions from symmetry-equivalent Pauli strings in the same orbit $\tilde{c}_{\alpha} = \sum_{j=0}^{n-1} \Tr[\hat{T}^j(\tilde{s}_{\alpha})U^\dagger Z_{\lceil n/2 \rceil} U]$ (see \figref{fig:fig1}\textbf{b}). 

In \figref{fig1:symmetry_Paulis}, we numerically verify that $N > \tilde{N}$ for all circuit layers, therefore leading to a reduced number of Pauli strings throughout the simulations as discussed in \equref{eq:O-merging}. The memory saving increases with the circuit layers as more Paulis are generated. The circuit depth necessary to generate Pauli dynamics exploring the full operator space increases with the number of qubits $n$, as seen by the shifted elbows. The reduction in the number of Paulis needed to be propagated saturates to the theoretical limit in \exref{ex:Zn} at high enough layers. In practice in order for Pauli propagation to be numerically feasible, the number of Paulis generated are typically truncated so that the evolved operator does not span the full operator space. To tackle the dynamics of a more difficult dynamics, we adopt numerical truncations in symmetry PP in the next example.

\textit{Example: all-to-all XXZ Heisenberg model.}
To demonstrate how symmetry PP allows us to simulate physically relevant dynamics, we investigate the dynamics of the all-to-all interacting Heisenberg model with long-range power law interactions decaying with a distance defined on a $2$-dimensional lattice over an exponent $\alpha$. 
This model is a rich playground, as the ground state phase diagrams exhibit a complex structure in entanglement and critical phenomena~\cite{frerot2017entanglement}. Moreover, it is also relevant for quantum simulations and metrology in atom-based platforms such as the Rydberg atom arrays, dipolar molecules and trapped ion arrays~\cite{perlin2020spin, bilitewski2021dynamical, muleady2023validating}. The power-law XXZ Hamiltonian is given by 
\begin{equation}\label{eq:hamXXZ}
    H_{\text{XXZ}} = - J_{\perp}\sum_{ i < j  } \frac{X_i X_j + Y_i Y_j + (\Delta + 1) Z_i Z_j}{|\mathbf{r}_{i} - \mathbf{r}_{j}|^{\alpha}},
\end{equation}
where $J_{\perp}$ is the interaction strength and $\Delta$ is the relative anisotropic coupling strength and $\mathbf{r}_{i}$ is the position of the $i= 1, \cdots, n$-th atom. For $\Delta=0$, the Hamiltonian reduces to the invariant model with $SU(2)$ spin rotational symmetry. In the absence of spin-orbit coupling $\alpha=0$, the Hamiltonian also has rotation and permutation symmetry and the dynamics can continually be simulated classically.

The unitary is the trotterized Hamiltonian evolution $U = \prod_{l=1}^{L} \prod_{ i <j  } e^{- i h_{ij} \delta}$ where the time step is $\delta = t / L$ and $h_{ij} = - J_{\perp} \left(X_i X_j + Y_i Y_j + (\Delta + 1) Z_i Z_j\right) /|\mathbf{r}_{i} - \mathbf{r}_{j}|^{\alpha}$ is the local Hamiltonian operator. It is possible to merge after only evolving with sublayers to further improve the efficiency. Because pairs of Pauli operators commute with each other e.g. $\commu{X_i X_j}{Y_i Y_j} = 0$, we can further decompose the rotation into Pauli rotations as $e^{i \left(\theta_x X_i X_j + \theta_y Y_i Y_j + \theta_z Z_i Z_j \right)} = e^{i \theta_x X_i X_j} e^{i \theta_y Y_i Y_j} e^{i \theta_z Z_i Z_j}$. Alternative strategies include higher order trotterization with multi product formula~\cite{childs2021theory}, finding a minimum set of grouped operators~\cite{berry2007efficient} and randomized trotterization~\cite{campbell2019random, childs2019faster}. In the present example, we have numerically observed that the trotterization order affects the relevant observable expectation values minimally. 

We will study the total spin operator~\cite{muleady2023validating}
\begin{equation}\label{eq:s2}
    \mathbf{S}^2 = \sum_{\mu = x, y, z} (\sum_{i=1}^{n} \sigma_{\mu, i})^2,
\end{equation}
where $\sigma_{\mu, i}$ are spin-half Pauli matrices (e.g. $\sigma_{x, i} = \frac{1}{2} X_i$).
The measurement operator and unitary evolution are translationally invariant, therefore we apply our merging algorithm with translation symmetries along both $x$ and $y$ directions independently. {While there exists other specialized methods such as discrete truncated Wigner approximation~\cite{wurtz2018cluster}, this dynamics is highly non-trivial for MPS-based simulations~\cite{muleady2023validating} owing to the all-to-all connected circuit topology. In our naive MPO simulations, we observe difficulty in convergence for simulation time $t \bar{J} \abs{\Delta}>0.02$ using a bond dimension up to $400$.}

In \figref{fig2:xxz_translation}, we demonstrate that symmetry PP allows us to simulate the dynamics more accurately than standard PP. To simulate effects of presence of noise, our numerical simulations include a layer of depolarizing noise with strength $\gamma$ after each trotterization step. The noise layer can be viewed as a time discretized model for a local Lindbladian~\cite{breuer2002theory, lidar2019lecture}. In addition, we adopt a coefficient truncation strategy such that only terms with coefficients greater than a truncation threshold $\epsilon$ are kept~\cite{rudolph2025pauli}. For instance, a truncated operator is given by $O^{\text{trunc}} = \sum_{\alpha \,| \,c_\alpha \leq \epsilon} c_\alpha s_\alpha$ for some small $\epsilon$. We study the effect of different truncation thresholds on simulation accuracy and find that symmetry PP achieves faster convergence even for a larger truncation threshold $\epsilon = 10^{-3}$ (circles) compared to a truncation threshold at $\epsilon = 5 \times 10 ^{-5}$ (solid line). In comparison, standard PP diverges at earlier times (dashed vs.~solid lines). Moreover, symmetry PP with a higher truncation threshold $\epsilon=0.05$ agrees with standard PP at a much lower threshold ($\epsilon=0.001$). This can be understood by realizing that symmetry merging tends to increase the magnitudes of the coefficients, in turn increasing the Pauli strings that are kept per threshold. While it is inevitably exponentially costly to simulate long-time dynamics, given a fixed budget of memory and computation time, symmetry PP converges to longer time and leads to more accurate results hence reducing the computational resources for a target accuracy. 

\section{Discussion}
We have established a group-theoretic symmetry-merging framework for Pauli propagation, called symmetry PP. By propagating only a set of Pauli orbit representatives,  we reduce the memory cost by a factor related to the size of the group. Such a memory saving is highly impactful for future improvements using graphics processing unit (GPU) because the acceleration hardwares typically have more constraints in memory.  Most recently, related works have also taken symmetry into consideration via merging in PP~\cite{loizeauQuantumManybodySimulations2024, haghshenas2025digital}. Theoretically, we provide rigorous guarantees for the correctness and computational complexity. Then we demonstrate the effectiveness of symmetry PP for an ergodic dynamics of $1d$ tilted-field Ising model and an all-to-all Heisenberg dynamics.  We observe an improved stability in symmetry PP simulations, despite subtleties in the memory saving due to an unavoidable exponentially large space of Pauli strings. It would be interesting to analyze the improvement by considering magic monotones in resource theory~\cite{dowling2025magic, turkeshi2025magic}.

For future work, it would be interesting to explore symmetric merging strategies in other propagation based methods such as Majorana propagation for fermionic systems~\cite{miller2025simulation,alam2025fermionic,alam2025programmable,danna2025majorana} or displacement propagation for bosonic system~\cite{upreti2025quantum}.
This includes any future merging strategies for other discrete symmetries.
More ambitiously, developing propagation subroutines 
for continuous symmetries will be highly relevant and may yield substantial reductions in computational cost. An abstract analysis of this direction has already been initiated using representation theory in Ref.~\cite{cirstoiu2024fourier}. Such progress would open new avenues for investigating the role of symmetries in physically relevant quantum Floquet dynamics, such as those previously explored in random circuits with $U(1)$ symmetry~\cite{khemaniOperatorSpreadingEmergence2018, rakovszkyDiffusiveHydrodynamicsOutTimeOrdered2018}.

\vspace{3mm}

\noindent \textit{Note}: While preparing this manuscript we used similar methods with our collaborators in \refcite{haghshenas2025digital}. 
\vspace{2mm}

\noindent \emph{Code and data availability.}
The numerical experiments in this work were conducted with the \href{https://github.com/MSRudolph/PauliPropagation.jl}{PauliPropagation.jl} package, which now also contains our symmetry merging functionality. The numerical data is publicly available at \href{https://doi.org/10.5281/zenodo.17804051}{Zenodo}.

\vspace{2mm}

\noindent \textit{Acknowledgments.} 
The authors acknowledge Nathan Leitao for insightful discussions. 
YT acknowledges support from NCCR spin, a National Centre of Competence in Research, funded by the Swiss National Science Foundation (grant number 565785). SYC was supported by Laboratory Directed Research and Development (LDRD) program of Los Alamos National Laboratory (LANL) under project number 20260043DR and by the U.S. Department of Energy, Office of Science, Office of Advanced Scientific Computing Research under Contract No. DE-AC05-00OR22725 through the Accelerated Research in Quantum Computing Program MACH-Q project. MSR acknowledges funding from the 2024 Google PhD Fellowship and the Swiss National Science Foundation [grant number 200021-219329]. ZH acknowledges support from the Sandoz Family Foundation-Monique de Meuron program for Academic Promotion. We acknowledge support of the NCCR MARVEL, a National Centre of Competence in Research, funded by the Swiss National Science Foundation (grant number 205602).

\let\oldaddcontentsline\addcontentsline
\renewcommand{\addcontentsline}[3]{}  
\bibliography{references}
\let\addcontentsline\oldaddcontentsline 


\onecolumngrid  
\appendix

\begin{appendix}
\clearpage 

\vspace{2.0em}
\begin{center}
\textbf{\Large Supplementary Materials}
\end{center}

\renewcommand{\appendixname}{}

\renewcommand{\thesubsection}{\MakeUppercase{\alph{section}}.\arabic{subsection}}
\renewcommand{\thesubsubsection}{\MakeUppercase{\alph{section}}.\arabic{subsection}.\alph{subsubsection}}
\makeatletter
\renewcommand{\p@subsection}{}
\renewcommand{\p@subsubsection}{}
\makeatother

\renewcommand{\figurename}{Supplementary Figure}
\setcounter{secnumdepth}{3}
\makeatletter
     \@addtoreset{figure}{section}
\makeatother

\vspace{3mm}

\tableofcontents

\section{Pauli Path Methods}\label{app:pp}
    In this appendix, we review the framework of Pauli path based simulations, adopting the conventions of Pauli propagation~\cite{rudolph2025pauli}.
    Let us consider an approximately maximally mixed state, commonly studied in operator spreading,
    \begin{equation}
        \rho = \frac{\mathbb{I} + a O}{\Tr(\mathbb{I})},
    \end{equation}
    which is perturbed by a small amplitude $a$.
    The evolution of the state is determined by the evolution of the perturbation operator
    \begin{equation}
        O^\prime = U^\dagger O U.
    \end{equation}
    Here, the operators live in the $n$-qubit Hilbert space $\mathcal{H} = (\mathbb{C}^2)^{\otimes n}$, with an inner product $\inner{A}{B} = \Tr[A^\dagger B]$. The operators can be represented by matrices $M_{2^n}(\mathbb{C})$.
    To analyze quantum dynamics, it is useful to vectorize these operators in an enlarged Hilbert space, called the Liouville space $\mathcal{L} = \mathcal{H} \otimes \mathcal{H}^*$~\cite{gyamfi2020fundamentals}. In the vectorization formalism, an operator $A$ on $\mathcal{H}$ can be uniquely mapped to a vector $\kett{A} \in \mathcal{L}$ as the following:
    \begin{equation}
        A = \sum_{ij} A_{ij} \ket{i}\bra{j} \quad  \longrightarrow \quad \kett{A} = \sum_{ij} A_{ij} \ket{i} \otimes \ket{j^*},
    \end{equation}
    where $\{\ket{i}\}$ represents an orthonormal basis in the Hilbert space. $\ket{j^*}$ is the complex conjugation of $\ket{j}$, corresponding to the dual vector $\bra{j}$. In the simple cases of the computational basis, the complex conjugation does nothing and is often abbreviated in vectorization formalism. The unitary acting on the original operators are superoperators in the Liouville space. In particular, its adjoint action is 
    \begin{equation}
        U \otimes U^* \kett{A} = \kett{U A U^\dagger}.
    \end{equation}
    
    In the rest of the sections, it is useful to define the \emph{normalized} Pauli strings
    \begin{equation}
        \mathcal{P}_n = \left\{\frac{\mathbb{I}}{\sqrt{2}}, \frac{X}{\sqrt{2}}, \frac{Y}{\sqrt{2}}, \frac{Z}{\sqrt{2}}\right \}^{\otimes n}.
    \end{equation}
    Consider the basis given by the normalized Pauli operators, $\kett{A}$ can be represented as 
    \begin{equation}
        \kett{A} = \sum_{s \in \mathcal{P}_n} \Tr[s^\dagger A] \, \kett{s}.
    \end{equation}
    In the following, we will often interchangably write $A = \sum_{s \in \mathcal{P}_n} \Tr[s A] \, s$ because $s \in \mathcal{P}_n$ forms a Hermitian basis. 
    The adjoint action of the unitary can then be represented in the Pauli basis as
    \begin{equation}\label{eq:liouville}
         U \otimes U^* = \sum_{ij} (U \otimes U^*)_{ij}\,\kett{s_i}\braa{s_j}, \quad  (U \otimes U^*)_{ij} = \Tr[Us_j U^\dagger s_i].
    \end{equation}
    The Liouville matrix representation $(U \otimes U^*)_{ij}$ in the Pauli basis is often referred to as the \textit{Pauli transfer matrix} (PTM) in Pauli propagation. While we focus on unitary evolutions, our analysis in this formalism can be generalized to quantum channels. The PTM for any quantum channels in the Pauli basis is always real. For instance in the case of a unitary evolution, its complex conjugate is the same in the Pauli basis gives
    \begin{equation}
        (U \otimes U^*)^*_{ij} = \Tr[(U^\dagger)^T (s^\dagger_j)^T U^T (s^\dagger_i)^T] = \Tr[s_i U s_j U^\dagger] = (U \otimes U^*)_{ij},
    \end{equation}
    where we have used that Pauli operators are Hermitian $s_i^\dagger = s_i$ as well as cyclic and invariant properties of trace $\Tr[(AB)^T] = \Tr[AB] = \Tr[BA]$. In in the Heisenberg picture, the PTM is given by the transpose of that in the Schrodinger picture
    \begin{equation}\label{eq:ptm-heisenberg}
        (U^+ \otimes U^T)_{ij} = \Tr[U^{\dagger }s_j U s_i]
    \end{equation}
    
    Suppose the unitary circuit is composed of $L$ layers as $U = U_L U_{L-1} \cdots U_2 U_1$, then the evolved expectation of an operator in the Heisenberg picture is 
    \begin{equation}
        U^{\dagger} O U = U_1^{\dagger} U_2^{\dagger} \cdots U_{L-1}^{\dagger} U_L^{\dagger} (O) U_L U_{L-1} \cdots U_2 U_1.
    \end{equation}
    The Pauli decomposition of the initial observable is
    \begin{equation}
        O_L := O = \sum_{s_L \in \mathcal{P}_n} \Tr[O s_L] s_L.
    \end{equation}
    For layers $l=L, L-1, \cdots, 1$-th layer, a partially evolved observable is given by recursively applying each layer of the unitaries
    \begin{equation}
        O_{l-1} = U_l^{\dagger} O_l U_l
    \end{equation}
    For instance, the last layer of the evolution can be written as a linear combination of applying the PTM from \equref{eq:ptm-heisenberg} to $s_L$ as
    \begin{equation}
        O_{L-1} := U_L^{\dagger} (O_L) U_L =  \sum_{s_L, s_{L-1} \in \mathcal{P}_n} \Tr\left[O s_{L}\right] \Tr \left[ U_L^{\dagger} s_{L} U_L s_{L-1} \right] s_{L-1}.
    \end{equation}
    Expanding such evolution iteratively in the Pauli basis and denoting the final operator as $O_0$ leads to the \emph{Pauli path integral}
    \begin{equation}\label{eq:Pauli_path}
    O_0 := U^{\dagger} O U = \sum_{s_0, s_1, \cdots s_{L-1} s_L \in \mathcal{P}_n} \Tr[s_L O] \prod_{l=1}^{L} \Tr[s_{l-1} U_l^{\dagger} s_{l} U_l] s_0.
    \end{equation}
    In the rest of the appendices, our goal is to improve the computation using symmetry analysis. 

\newpage
\section{Symmetry and Group Theory}\label{app:group}

In this self-contained section, we review the terminologies that we use from group theory, which can be found in standard texts~\cite{armstrong1997groups, rotman2012introduction}. 
\begin{definition}[\textbf{Group} (finite)]
    A finite set $G$ with an operation $*$ forms a \emph{group} $(G, *)$ if the following holds:
    \begin{enumerate}
        \item The identity element is in the group: $ \exists e  \in G, \quad  e*g = g*e = g, \quad \forall g \in G$.
        \item Every element has an inverse: $\forall g \in G, \exists a \in G \implies g*a = a*g = e$.
        \item The operation $*$ is associative: $g * (a * b) = (g * a) * b, \quad \forall g, a, b \in G$.
    \end{enumerate}
\end{definition}
In the following, we will usually omit the operation and refer to a group as $G$. Rather than focusing on the abstract definition of a group, it is more useful to think about their \emph{action} on a set. 
\begin{definition}[\textbf{Group action}]
    The action of a group $G$ on a set $X$ is denoted as
    \begin{equation}
        g \cdot (x), \quad \forall g \in G, \, x\in X.
    \end{equation}
    The group action defines an equivalence relation on $X$: $x \sim y$ if $\exists\, g \in G$ such that
    \begin{equation}
        y = g\cdot(x).
    \end{equation}
\end{definition}

The set of elements reached by acting $G$ on an element $x$ is called an \textit{orbit}.
\begin{definition}[\textbf{Group orbit}]\label{def:orbit}
    The orbit of $x\in X$ under the action of group $G$ is 
    \begin{equation}
        \operatorname{Orb}_G(x) := \{g\cdot(x) \, \vert \, g\in G \}.
    \end{equation}
    Each orbit $\operatorname{Orb}_G(x)$ is an \emph{equivalence class} of $x$, and two elements $x, y$ are equivalent under the group action if they are in the same orbit $y \in \operatorname{Orb}_G(x)$.
\end{definition}

Because each orbit is a set containing elements that are reachable by the action of a group, in this sense these elements in an orbit are redundant descriptions of each other. Therefore it is sufficient to only consider a representative for each orbit. 
\begin{definition}[\textbf{Orbit representatives}]\label{def:orbit-repr}
    A complete set of orbit representatives is a set $R^G_X :=\{x_i\} \subseteq X$ such that
    \begin{equation}
        X = \coprod_{x_i \in R^G_X} \operatorname{Orb}_G(x_i),
    \end{equation}
    where $\coprod$ denotes a disjoint union.
\end{definition}
Here the disjoint union emphasizes that different orbits do not intersect or otherwise they are the same orbit. 
Applying a group action on Pauli strings, we define \emph{Pauli representatives} under translation symmetry. 
\begin{example}[\textbf{Pauli representatives ($\mathbb{Z}_n$)}]\label{ex:Pauli-orbits-zn}
    Let $\mathcal{P}_n$ be the set of normalized Pauli strings, and $\mathbb{Z}_n$ be the symmetry group consisting of $n$ group elements $\mathbb{Z}_n = \{\hat{I}, \hat{T}, \cdots, \hat{T}^{n-1} \}$. The group action of $\mathbb{Z}_n$ on Pauli string $s \in \mathcal{P}_n$ is
    \begin{equation}
        \operatorname{Orb}_{\mathbb{Z}_n}(s) = \sum_{j = 0}^{n-1} \hat{T}^j \cdot (s).
    \end{equation}
    Define the action of translation as shifting the system to the left such that $\hat{T}(IZI\cdots I) := ZII\cdots I$ (up to normalization). Let $s_0 = 1/\sqrt{2}^n (ZII\cdots I)$, then its orbit consists of all the Pauli strings under translation 
    \begin{equation}
        \operatorname{Orb}_{\mathbb{Z}_n}(s_0) = \left\{\frac{ZII\cdots I}{\sqrt{2}^n},  \frac{IZI\cdots I}{\sqrt{2}^n},  \cdots, \frac{III\cdots Z}{\sqrt{2}^n} \right\}.
    \end{equation}
    A complete set of orbits partition the Pauli strings into a disjoint set. For instance, for $n=3$ qubits, the set of unique orbits representatives (up to normalization) are
    \begin{align}
        \{&III, XII, YII, ZII, XXI, XYI, XZI, YXI, YYI, YZI, ZXI, ZYI, ZZI, XXX, XXY, XXZ, \nonumber \\
        & YYY, YYX, YYZ, ZZX, ZZY, ZZZ, XYZ, XZY\}.
    \end{align}
    In total, there are $24$ representatives and we show how to calculate its cardinality in \exref{ex:Zn_app}.
\end{example}
Next, we formalize how to count the size of the orbits of any groups using \lemref{lem:burnside}.
\begin{lemma}[\textbf{Burnside's lemma}~\cite{armstrong1997groups}]\label{lem:burnside}
    Let $G$ be a finite group, and $R^{G}_X$ be a complete set of orbit representatives of a finite set $X$. Then the cardinality of the representatives is given by
    \begin{equation}
        \abs{R^{G}_X} = \frac{1}{\abs{G}} \sum_{g\in G} \abs{X^g},
    \end{equation}
    where $X^g$ denotes the set of elements that are invariant under left action of $g$
    \begin{equation}
        X^g = \{x \in X \, \vert \, g \cdot(x) = x\}.
    \end{equation}
\end{lemma}
\begin{proof}
    Consider the sum in the r.h.s.
    \begin{equation}
        \sum_{g\in G} \abs{X^g} = \abs{\{ (x, g) \, \vert \,  x \in X, g\in G \, , \, g \cdot(x) = x\}}.
    \end{equation}
    The sum can be rewritten as
    \begin{equation}\label{eq:fixed_point_stab}
         \sum_{g\in G} \abs{X^g} = \sum_{x \in X} \abs{\operatorname{Stab}_G(x)},
    \end{equation}
    where $\operatorname{Stab}_G(x) = \{g \,|\, g\cdot(x) = x, \, g\in G\}$ is the stabilizer of $x$. Consider an element $x\in X$, the orbit-stabilizer theorem counts how much $x$ is stabilized and moved in orbits  
    \begin{equation}
        \abs{G} = \abs{\operatorname{Orb}_G(x)} \times \abs{\operatorname{Stab}_G(x)},
    \end{equation}
    leading to 
    \begin{equation}
         \sum_{x \in X} \abs{\operatorname{Stab}_G(x)} = \sum_{x \in X} \frac{\abs{G}}{\abs{\operatorname{Orb}_G(x)}}.
    \end{equation}    
    By definition of orbit representatives, $X$ is the disjoint union of orbits so the sum is
    \begin{align}
        \sum_{x \in X} \abs{\operatorname{Stab}_G(x)} &= \sum_{r \in R^{G}_X}  \sum_{x \in \operatorname{Orb}_G(r)} \frac{\abs{G}}{\abs{\operatorname{Orb}_G(x)}}, \\
        & =\abs{G} \abs{R^{G}_X}.
    \end{align}
    Therefore, combing with \equref{eq:fixed_point_stab}, we have arrived at the result
    \begin{equation}
        \sum_{g\in G} \abs{X^g} = \abs{G} \abs{R^{G}_X}.
    \end{equation}
\end{proof}

Equipped with \lemref{lem:burnside}, we can then count the representatives under group actions for Pauli strings. 
\begin{corollary}[\textbf{Pauli representatives cardinality}~\cite{armstrong1997groups}]\label{cor:burnside_Pauli}
    Let $G$ be a finite group, and $\mathcal{R}^{G}_n$ be a complete set of orbit representatives of the normalized $n$-qubit Pauli strings $\mathcal{P}_n$. Then the cardinality of the representatives is given by
    \begin{equation}
        \abs{\mathcal{R}^{G}_n} = \frac{1}{\abs{G}} \sum_{g\in G} \abs{\mathcal{P}_n^g},
    \end{equation}
    where $\mathcal{P}_n^g$ denotes the set of Paulis that are invariant under left action of $g$
    \begin{equation}
        \mathcal{P}_n^g = \{s \in \mathcal{P}_n \, \vert \, g \cdot(s) = s\}.
    \end{equation}
\end{corollary}
In the following examples, we reserve the notation $\mathcal{R}^G_n$ for the $n$-qubit \emph{Pauli representatives} under a group $G$. 

\begin{example}[\textbf{Translation symmetry} $\mathbb{Z}_n$]\label{ex:Zn_app}
 Consider the translational symmetry group $G=\mathbb{Z}_n$ on $n$ qubits. The cardinality of the group is $\abs{\mathbb{Z}_n} = n.$ According to \corref{cor:burnside_Pauli}, the number of orbit representatives is given by computing the size of fixed Paulis under the group action
 \begin{equation}
     \abs{\mathcal{R}^{\mathbb{Z}_n}_n}
 = \frac{1}{\abs{\mathbb{Z}_n}} \sum_{j = 0}^{n-1}   \abs{\mathcal{P}_n^{\hat{T}^j}}.
 \end{equation}
 Mathematically, this is a classical combinatorics problem of counting \emph{$p$-ary necklaces} with $n$ beads, where each bead can be one of $p$ colors (for counting Pauli strings $p=4$ corresponding to $\{I, X, Y, Z\}$).
 The number of distinct \emph{necklace} combinations is given by
\begin{equation}
\abs{\mathcal{R}^{\mathbb{Z}_n}_n}
 = \frac{1}{n} \sum_{j = 1}^{n} 4^{\operatorname{gcd}(j, n)},
\end{equation} 
where $\operatorname{gcd}(a,b)$ represents the greatest common divisor between two integers $a$ and $b$. Here, $j=n$ corresponds to applying the identity element and is equivalent to counting the size of the Pauli strings $\abs{\mathcal{P}_n} = 4^n$. For $1 \leq j<n$, a length-$n$ Pauli string invariant under shift of length $j$ has periodicity of $\operatorname{gcd}(j, n)$. Moreover, there are $4^{\operatorname{gcd}(j, n)}$ number of such strings given by the first periodic subset and the rest of the necklace is fully fixed. 
In particular, for $n = 2^m$ , this can be simplified into 

\begin{align} 
  \abs{\mathcal{R}^{\mathbb{Z}_n}_n} 
  & = \frac{1}{n} \sum_{j = 1}^{n} 4^{\textrm{gcd}(j, n)}, \nonumber \\ 
  &=  \frac{1}{n}  \left(4\cdot\frac{n}{2} + 4^2 \cdot \frac{n}{4} + 4^4 \cdot \frac{n}{8} + \cdots +  4^{n/2} + 4^n \right), \nonumber\\ 
                & = \frac{1}{n}  \left( \sum_{k=1}^m 4^{2^{k-1}}\cdot 2^{m-k} + 4^n \right) \le 2 \cdot \frac{4^n}{n}\;, 
\end{align}
as the number of indices $j$ with $\textrm{gcd}(j, n) = 2^{k-1}$ for $k = 1,\cdots, m $ is $\frac{n}{2^k}$. 

On the other hand, if $n$ is a prime number, then, 
\begin{equation}
  \abs{\mathcal{R}^{\mathbb{Z}_n}_n} = \frac{1}{n}(4^n + 4(n-1))).
\end{equation}
Therefore by \thmref{thm:space-complexity}, we have a linear saving ratio $r \in \Theta(\frac{\abs{\mathcal{R}^{\mathbb{Z}_n}_n}}{4^n})\sim\Theta(\frac{1}{n})$ using symmetry Pauli propagation. 
\end{example}
        
\begin{example}[\textbf{Permutation symmetry} $S_n$]
 Consider the permutation symmetry group $S_n$. The number of orbit representatives can be found using the stars and bars formula, leading to: 
\begin{equation}
 \abs{\mathcal{R}^{S_n}_n} = \binom{n+3}{3} : = \textrm{Te}_{n+1} \;. 
\end{equation} 
This is also called the tetrahedral number $\textrm{Te}_{n} = \frac{n (n+1) (n+2)}{6}$~\cite{schatzki2024theoretical}.
For large $n$, the number of orbit representatives scales as $\Theta(n^3)$, leading to $r \in \Theta\left(\frac{n^3}{4^n}\right)$  improvement in terms of space complexity.
\end{example}

\begin{example}[\textbf{Dihedral group} $D_n$] The \emph{Dihedral group} $D_n$ corresponds to a group of the reflectional and rotational symmetry of an regular $n$-sided polygon.  By definition, $D_n$ consists of $2n$ elements: $n$ rotational symmetries and $n$ reflectional symmetries.
Therefore, we have a linear time saving $r \in \Theta(\frac{1}{2n})$ using symmetric merging. This can be viewed as a natural extension of translational symmetry with periodic boundary condition (c.f. Example~\ref{ex:Zn}), where rotational operations on the vertices of the polygon are interpreted as cyclic translations along its perimeter.

\end{example}

\newpage
\section{Symmetry-merging Pauli Propagation}\label{app:symmetry-merging}
In this section, we prove the correctness of symmetry Pauli propagation under exact evolution. First, we turn the intuition that evolving a Pauli string is the same as evolving its orbit representative into a rigorous statement via \lemref{lem:expectation-equib}.
\begin{lemma}[\textbf{Expectation equivalence for orbit representatives}]\label{lem:expectation-equib}
    Let $G$ be the symmetry group and $U$ be a symmetric unitary. Suppose a function over Pauli strings $f: \mathcal{P}_n \rightarrow \mathbb{R}$ is also invariant under $G$ such that $f(g\cdot (s)) = f(s), \, \forall g \in G$. Then the expectation value after propagating a Pauli string $s_1 \in \mathcal{P}_n$ is equivalent to that from propagating any representative in the same orbit $g\cdot (s_1), \, g \in G$. More precisely the following holds:
    \begin{equation}
        \sum_{s_0 \in  \mathcal{P}_n}\Tr[U  g\cdot(s_1) \,U^\dagger s_0] f(s_0) = \sum_{s_0 \in  \mathcal{P}_n} \Tr[U s_1 U^\dagger s_0] f(s_0), \quad \forall g \in G.
    \end{equation}
\end{lemma}

\begin{proof}
Let $A_g$ be the representation of the the symmetry group so that $\commu{A_g}{U} = 0$. The action of $g$ on $s_1$ can be written as its inverse action on $s_0$ as
    \begin{align}\label{eq:equiv-exp}
        \Tr[U  g\cdot(s_1) \,U^\dagger s_0] &= \Tr[U  A_g s_1 \nonumber A_g^{(-1)} U^\dagger s_0], \\
        & = \Tr[U  s_1 U^\dagger A_g^{(-1)} s_0 A_g ], \nonumber \\
        & = \Tr[U  s_1 \,U^\dagger g^{(-1)}\cdot(s_0) ].
    \end{align}
Substituting the summation over the Pauli group with its orbit representatives, the expectation value we want to evaluate can be rewritten as
\begin{equation}\label{eq:Pauli-orbits}
\sum_{s_0 \in  \mathcal{P}_n}\Tr[U  g\cdot(s_1) \,U^\dagger s_0] f(s_0) = \sum_{g_0 \in G} \sum_{s_0 \in  \mathcal{R}^{G}_n} \Tr[U s_1 U^\dagger g_0 \cdot(s_0)] f(g_0 \cdot(s_0)), \quad \forall g \in G.
\end{equation}
Using \equref{eq:equiv-exp} and \equref{eq:Pauli-orbits}, we have
\begin{equation}
    \sum_{s_0 \in  \mathcal{P}_n}\Tr[U  g\cdot(s_1) \,U^\dagger s_0] f(s_0)  = \sum_{g_0 \in G} \sum_{s_0 \in  \mathcal{R}^{G}_n} \Tr[U  s_1 \,U^\dagger g^{(-1)}g_0 \cdot(s_0)] f(s_0), 
\end{equation}
where we have also used that $f$ is invariant under the symmetry transformation. Therefore, we can relable the sum over the group element by $\tilde{g}_0 = g g_0$ so that
\begin{equation}
    \sum_{g_0 \in G} \sum_{s_0 \in  \mathcal{R}^{G}_n} \Tr[U  s_1 \,U^\dagger g^{(-1)}g_0 \cdot(s_0)] f(s_0) = \sum_{\tilde{g}_0 \in G} \sum_{s_0 \in  \mathcal{R}^{G}_n} \Tr[U  s_1 \,U^\dagger \tilde{g}_0 \cdot(s_0)] f(s_0).
\end{equation}
Using the fact that $\mathcal{R}^{G}_n$ is a complete set of orbit representatives, the r.h.s becomes
\begin{equation}
    \sum_{\tilde{g}_0 \in G} \sum_{s_0 \in  \mathcal{R}^{G}_n} \Tr[U  s_1 \,U^\dagger \tilde{g}_0 \cdot(s_0)] f(s_0) = \sum_{s_0 \in  \mathcal{P}_n} \Tr[U s_1 U^\dagger s_0] f(s_0).
\end{equation}
\end{proof}

Using \lemref{lem:expectation-equib}, we provide a proof of \thmref{thm:symmetry} by induction. 
\setcounter{theorem}{0}
\begin{theorem}[\textbf{Pauli Symmetric Merging}]
    Let $\rho$ be the initial state and $U = \prod_{l=1}^L U_l$ be an $L$-layer unitary. Suppose they have symmetry $G$ with representations $A_g$ such that 
    \[
    \commu{A_g}{\rho} = \commu{A_g}{U_l}=0, \quad \forall g \in G.
    \]
    Let $\Tr[\rho\, U^\dagger O U]$ be the target expectation value of an observable $O$. Then the output value given by the symmetry-merging propagation algorithm via Algorithm~\hyperref[alg:symmetry]{1} 
    \begin{equation}
        \widetilde{O}_0 = \sum_{s_0, \dots, s_{L} \in \mathcal{R}^{G}_n} \sum_{g_0, \dots, g_{L} \in G} \, \prod_{l=1}^{L} \Tr[U_l^\dagger s_l U_l g_{l} \cdot (s_{l-1})] \Tr[\rho s_0],
    \end{equation}
    equals the target value:
    \begin{equation}
        \Tr[\rho \, \widetilde{O}_0] = \Tr[\rho\, U^\dagger O U].
    \end{equation}
\end{theorem}

\begin{proof}
We prove the theorem by \textit{induction} on the number of layers $l$. For conceptual simplicity of the proof, we consider the merging process from the \textit{beginning} layers of the unitary circuits, even though in Heisenberg evolution we perform propagation from the final layers.

\textbf{Base Case:} (\( m = 2 \))  
Consider the expectation value after evolution of $L$ layers of the unitary:
\begin{equation}
     \Tr[\rho\, U^\dagger O U] = \Tr\left[\rho \left( \prod_{l=1}^{L} U_l^\dagger \right) O \left( \prod_{l=1}^{L} U_l \right) \right].
\end{equation}
Expanding in the Pauli basis leads to the conventional Pauli path integral
\begin{equation}
    \sum_{s_0, \dots, s_L \in \mathcal{P}_n} \Tr[O s_L] \prod_{l=1}^{L} \Tr[U_l^\dagger s_l U_l s_{l-1}] \Tr[\rho s_0].
\end{equation}
To consider symmetric merging, we rewrite the beginning two layers $l=1, 2$ of the unitary evolution in terms of orbit representatives and their symmetry group as
\begin{equation}
    \sum_{s_2, \dots, s_L \in \mathcal{P}_n} \Tr[O s_L] \prod_{l=3}^{L} \Tr[U_l^\dagger s_l U_l s_{l-1}] \sum_{s_1, s_0 \in \mathcal{R}^{G}_n} \sum_{g_1, g_0 \in G} \Tr[U_2^\dagger s_2 U_2 g_1\cdot(s_{1})] \Tr[U_1^\dagger g_1 \cdot(s_1) U_1 g_0\cdot(s_{0})] \Tr[\rho \, g_0\cdot(s_{0})].
\end{equation}
Because the initial state is symmetric under $G$, the Pauli strings in the same orbit contribute equally to the expectation value:
\begin{equation}
	\Tr[\rho \, g\cdot(s)] = \Tr[\rho s], \quad  \forall g \in G.
\end{equation}

Let $f(s_0) =  \Tr[\rho s_0]$, by \lemref{lem:expectation-equib} we can simplify the last two terms as:

\begin{equation}
    \sum_{s_2, \dots, s_L \in \mathcal{P}_n} \Tr[O s_L] \prod_{l=3}^{L} \Tr[U_l^\dagger s_l U_l s_{l-1}] \sum_{s_0, s_1 \in \mathcal{R}^{G}_n} \sum_{g_0, g_1 \in G} \Tr[U_2^\dagger s_2 U_2 g_1\cdot(s_{1})] \Tr[U_1^\dagger s_1 U_1 g_0 \cdot(s_{0})] \Tr[\rho \, s_{0}].
\end{equation}

The expressions of the form $\Tr[U_2^\dagger s_2 U_2 g_1\cdot(s_{1})] \Tr[U_1^\dagger s_1 U_1 g_0 \cdot(s_{0})] $ is saying that after Pauli propagation until layer $l=2$, we can simply propagate the orbit representatives in the last layer $l=1$.
We establish the merged form up the initial two layers.

\textbf{Inductive Step:}  
Suppose that for the beginning \( m-1 \) layers, we can merge Pauli strings into symmetry orbits. We show that this also holds for the \( m \)-th layer.

By the inductive assumption, propagating the orbit representatives up to the first $m-1$ layers gives an exact expectation value: 
\begin{align}
    &\Tr\left[ \rho \left( \prod_{l=1}^{L} U_l^\dagger \right) O \left( \prod_{l=1}^{L} U_l \right) \right] \nonumber \\
    = &\sum_{s_m, \dots, s_L \in \mathcal{P}_n} \prod_{l=m+1}^{L} \Tr[U_l^\dagger s_l U_l s_{l-1}] \sum_{s_0, \dots, s_{m-1} \in \mathcal{R}^{G}_n} \sum_{g_0, \dots, g_{m-1} \in G} \prod_{l=1}^{m} \Tr[U_l^\dagger s_l U_l g_{l} \cdot (s_{l-1})] \Tr[\rho s_0]. 
\end{align}
Now, rewriting the Pauli string in the next layer $l=m$ of the unitary evolution up to $m+1$ layers in terms of orbit representatives, the expectation value on the r.h.s. becomes
\begin{align}
	&\sum_{s_{m+1}, \dots, s_L \in \mathcal{P}_n} \prod_{l=m+2}^{L} \Tr[U_l^\dagger s_l U_l s_{l-1}] \sum_{s_m \in  \mathcal{P}_n}  \sum_{s_{m-1} \in \mathcal{R}^{G}_n} \sum_{g_{m-1} \in G} \Tr[U_{m+1}^\dagger s_{m+1} U_{m+1} s_m] \Tr[U_m^\dagger s_m U_m g_{m-1} \cdot(s_{m-1}) ] f(s_{m-1}) \nonumber \\
	= & \sum_{\substack{s_{m+1}, \dots, s_L \\ \in \mathcal{P}_n}} \prod_{l=m+2}^{L} \Tr[U_l^\dagger s_l U_l s_{l-1}] \sum_{\substack{s_{m-1}, s_m  \\ \in \mathcal{R}^{G}_n }} \sum_{\substack{g_{m-1}, g_m \\ \in G}}  
	\Tr[U_{m+1}^\dagger s_{m+1} U_{m+1} g_{m} \cdot(s_{m}) ] \, \Tr[U_m^\dagger g_m \cdot(s_m) U_m g_{m-1} \cdot(s_{m-1}) ] f(s_{m-1}), \label{eq:merging_proof_expec}
\end{align}

where $f$ is the following function 
\begin{equation}
	f(s_{m-1}) = \sum_{s_0, \dots, s_{m-2} \in \mathcal{R}^{G}_n} \sum_{g_0, \dots, g_{m-2} \in G} \prod_{l=1}^{m-1} \Tr[U_l^\dagger s_l U_l g_{l} \cdot (s_{l-1})] \Tr[\rho s_0].
\end{equation}
By the inductive assumption, $f$ is invariant under the symmetry action such that $f(g \cdot(s_{m-1})) = f(s_{m-1}), \, \forall g \in G$.

Therefore, applying the same symmetry merging argument using \lemref{lem:expectation-equib} for layer \( m \), leads to  the following: 
\begin{equation}
	\sum_{s_{m-1} \in \mathcal{R}^{G}_n} \sum_{g_{m-1} \in G} \Tr[U_m^\dagger g_m \cdot(s_m) U_m g_{m-1} \cdot(s_{m-1}) ] f(s_{m-1}) = 	\sum_{s_{m-1} \in \mathcal{R}^{G}_n} \sum_{g_{m-1} \in G}\Tr[U_m^\dagger s_m U_m g_{m-1} \cdot(s_{m-1}) ] f(s_{m-1}), \label{eq:merging_proof_mlayer}
\end{equation}

Using \equref{eq:merging_proof_mlayer}, we can simplify \equref{eq:merging_proof_expec} to
\begin{align}
& \sum_{\substack{s_{m+1}, \dots, s_L \\ \in \mathcal{P}_n}} \prod_{l=m+2}^{L} \Tr[U_l^\dagger s_l U_l s_{l-1}] \sum_{\substack{s_{m-1}, s_m  \\ \in \mathcal{R}^{G}_n }} \sum_{\substack{g_{m-1}, g_m \\ \in G}}  
	\Tr[U_{m+1}^\dagger s_{m+1} U_{m+1} g_{m} \cdot(s_{m}) ] \, \Tr[U_m^\dagger s_m U_m g_{m-1} \cdot(s_{m-1}) ] f(s_{m-1}), \nonumber \\
& =\sum_{\substack{s_{m+1}, \dots, s_L \\ \in \mathcal{P}_n}}  \prod_{l=m+2}^{L} \Tr[U_l^\dagger s_l U_l s_{l-1}] \sum_{s_0, \dots, s_{m} \in \mathcal{R}^{G}_n} \sum_{g_0, \dots, g_{m} \in G} \prod_{l=1}^{m} \Tr[U_l^\dagger s_l U_l g_{l} \cdot (s_{l-1})] \Tr[\rho s_0],
\end{align}
which says that we achieve the same expectation value after symmetric merging up to $m$ layers.

Thus, by induction, Pauli strings can be merged at every layer, completing the proof.
\end{proof}

\newpage
\section{Computational Complexity}\label{app:complexity}

In this section, we provide a rigorous analysis of computational cost of symmetry PP. First we provide a proof of \thmref{thm:space-complexity} in the main text.
\begin{theorem}[\textbf{Space complexity}]\label{thm:space-complexity-proof}
        
        Let $G$ be the discrete symmetry group acting on the set of $n$-qubit normalized Pauli strings $\mathcal{P}_n$. Suppose the observable at $l$-th layer has support on all of Pauli strings $\Tr[O_l s] \neq 0, \, \forall s\in \mathcal{P}_n$. The symmetry Pauli propagation defined in Algorithm~\hyperref[alg:symmetry]{1} outputs an observable $\widetilde{O}_l$ and reduces the space complexity by a ratio 
        \begin{equation}
            r = \frac{1}{4^n\abs{G} } \sum_{g\in G} \abs{\mathcal{P}_n^g} < 1,
        \end{equation} 
        where $\mathcal{P}^g_n = \{ g \cdot(s)\, | \, s \in \mathcal{P}_n\}$ corresponds to group action on Pauli strings. 
\end{theorem}
\begin{proof}
    The relative space complexity with respect to standard Pauli propagation is determined by the number of distinct operator trajectories that need to be evolved under symmetry Pauli propagation. The group $G$ acts on the normalized Pauli operators $\mathcal{P}_n$ through a left action $g\cdot (s)$, corresponding to symmetry transformations. For instance, consider the $3$-qubit Pauli group under translation group $\mathbb{Z}_3 = \{\hat{I}, \hat{T}, \hat{T}^2\}$. Then the action of $\hat{T}$ on a Pauli string is e.g. $\hat{T} \cdot (\frac{1}{\sqrt{2}^3} IZI) = \frac{1}{\sqrt{2}^3} IIZ$.
    This action partitions $\mathcal{P}_n$ into disjoint orbits:
    \begin{equation}
        \mathcal{P}_n = \coprod_{r_i \in \mathcal{R}^G_n} \operatorname{Orb}_G(r_i),
    \end{equation}
    where $\mathcal{R}^G_n$ is a complete set of orbit representatives (Definition~\ref{def:orbit-repr}) of the Pauli group.
    Each orbit corresponds to a set of Pauli strings related by the symmetry, so evolving one representative per orbit suffices by \lemref{lem:expectation-equib}. By assumption, $O_l$ has support on all Paulis and the merged observable using symmetry-merging algorithm is
    \begin{equation}
        \widetilde{O}_l = \sum_{s \in \mathcal{R}^G_n} s \left(\sum_{g\in G} \Tr\left[g\cdot(s) O_l\right]\right).
    \end{equation}
    Consider the case where $\widetilde{O}_l$ has support over all the representatives,  then the number of effective Pauli trajectories equals the number of distinct orbits,
    \begin{equation}
       \abs{\mathcal{R}^G_n} = \frac{1}{\abs{G}} \sum_{g\in G} \abs{\mathcal{P}_n^g},
    \end{equation}
    by Burnside’s lemma (\corref{cor:burnside_Pauli}).
    This quantity directly measures how many unique classes of Pauli operators remain after merging symmetry-equivalent Paulis. If the coefficients for some of the representatives are zero, then the support of the merged observable will be less than the size of the orbit representatives
    \begin{equation}
        \abs{\widetilde{O}_l} < \abs{\mathcal{R}^G_n}.
    \end{equation}
    
    In the absence of symmetry, $\abs{\mathcal{P}_n} = 4^n$, so naive Pauli propagation requires $\mathcal{O}(4^n)$ updates in the worst case. We can upper bound the relative support comparing the merged observable and the full observable being propagated as
    \begin{equation}
        \frac{\abs{\widetilde{O}_l}}{\abs{O_l}} \leq  \frac{\abs{\mathcal{R}^G_n}}{4^n}. 
    \end{equation}
    
     Therefore we define a \emph{ratio} $r$ that quantifies the fraction of operator space that are distinct under symmetry transformation as
\begin{align}
    r : & = \frac{\abs{\mathcal{R}^G_n}}{4^n},  \nonumber \\
    &= \frac{1}{\abs{G}} \frac{1}{4^n} \sum_{g\in G} \abs{\mathcal{P}_n^g}.
\end{align}

    Thus, symmetry-merged Pauli propagation reduces the exponential prefactor by the average group orbit size.
    The complexity improvement depends on the symmetry group and the sparsity of invariant Paulis under its action.
\end{proof}

We have formally considered the space complexity by counting the Pauli support in \thmref{thm:space-complexity-proof}, assuming that the observable $O_l$ has a full support in $\mathcal{P}_n$. In practice, the difference between merged and original observables depends on how the Paulis contributions are distributed in their orbits. In general, we may benefit from choosing the number of layers for merging as a hyperparameter. In this work, we consider the simplest scenario of utilizing the symmetry Pauli propagation for all the layers and we comment on scenarios where a more sophisticated design might be beneficial.
\begin{observation}
    Consider an observable that is supported exactly on the orbit representatives 
    \begin{equation}
        O = \sum_{s \in \mathcal{R}^G_n}  \Tr\left[s O\right] s.
    \end{equation}
    Then the symmetry Pauli propagation leaves the observable unchanged $\widetilde{O} = O$ and they have the same amount of Pauli support $\abs{\widetilde{O}} = \abs{O}$.
\end{observation}
This leads to the intuition that if an observable is sparsely distributed across different Pauli orbits, then it may not be beneficial to do symmetry-merging. Next we consider scenarios where it is highly beneficial to perform symmetry merging. 

\begin{observation}\label{obs:orbit-lower-bound}
    Suppose an observable is uniformly supported on Paulis contained in a single orbit represented by $s_0$ such that
    \begin{equation}
        O = \sum_{s \in \operatorname{Orb}_G(s_0)} \Tr[s O] s.
    \end{equation}
    By Lagrange's theorem, the size of an orbit is a divisor of the group size such that $\abs{G}/\abs{\operatorname{Orb}_G(s)} \in \mathbb{Z}^+, \, \forall s \in \mathcal{P}_n$. The relative support of a merged observable is then lower bounded by the size of the group
    \begin{equation}
        \frac{1}{\abs{G}} \leq \frac{\abs{\widetilde{O}_l}}{\abs{O_l}}.
    \end{equation}
    And the equality is achieved precisely when the size of the orbit is equal to the group size $\abs{\operatorname{Orb}_G(s)} = \abs{G}$. 
\end{observation}
Finding the largest orbit for an arbitrary group is an open question, to the best of our knowledge~\cite{qian2021large}, hence this lower bound is not necessarily achievable. Nevertheless, for most common symmetry groups appearing in the context of actions on Pauli strings in quantum simulations, our lower bound in \obsref{obs:orbit-lower-bound} is usually satisfied. In the rest of the section, we analyze the time complexity of our algorithm and comment on its relations to practical simulations. 

On a general level, there is a trade-off between the additional time spent performing merging and the time saved from propagating fewer Pauli operators. The more time we spent on merging, the more time we also save by needing to propagate fewer Paulis. We provide a theoretical bound on the time complexity aspects of symmetry PP in \propref{prop:time-complexity} as well as a numerical analysis of such time in practice in \supfigref{supfig1:ising_time}.

\setcounter{proposition}{0}
\begin{proposition}[\textbf{Time complexity}]\label{prop:time-complexity}
    Let $U=\prod_{l=1}^{L} U_l$ be an $l$-layered unitary symmetric under group $G$ and suppose each symmetric layer contains $K_l$ number of Pauli gates such that $U_l = \prod_{k=1}^{K_l} V_k$. Denote by $\abs{O_l}$ the number of Pauli terms in the operator expansion at layer $l$, and by $\abs{\tilde{O}_l}$ the corresponding number after symmetry merging.
    The total computational time of the symmetry Pauli propagation has two contributions as
    \begin{equation}
        T_{\mathrm{sym}} = \sum_{l=1}^{L} (T_{m,l} + T_{p,l}),
    \end{equation}
    where $T_{m,l}$ represents time spent merging Paulis and $T_{p, l}$ represents time spent on propagation. 
    More precisely, we have
    \begin{align}
        T_{m,l} &= \mathcal{O}(\abs{\tilde{O}_{l-1}}\abs{G}), \\
        T_{p,l} &= \mathcal{O}(2^{K_l}\abs{\tilde{O}_l}),
    \end{align}
    with $K_l$ the number of local gates in layer $l$ acting on distinct supports.
    
\end{proposition}

\begin{proof}

   We separate the analysis into two components: the merging cost $T_{m, l}$ and the propagation cost $T_{p,l}$.

    \paragraph{(1) Merging time complexity.}
    At layer $l$, the observable evolves as
    \begin{equation}
        O_{l-1} = U_l^{\dagger} \tilde{O}_l U_l,
    \end{equation}
    where $\tilde{O}_l$ denotes the post-merged operator from the previous layer.
    In the brute-force merging procedure (\appref{app:merging_exhaustive}), each Pauli in $\tilde{O}_{l-1}$ is checked against its symmetry orbit under $G$, requiring at most $\abs{G}$ group actions per element.
    Hence, the worst-case merging time is bounded by
    \begin{equation}
        T_{m,l} = \mathcal{O}(\abs{\tilde{O}_{l-1}}\abs{G}).
    \end{equation}

    \paragraph{(2) Propagation time complexity.}
    Each layer $U_l$ consists of $K_l$ Pauli gates $V_k$ (e.g., two-qubit gates) written as
    \begin{equation}
        V_k = e^{i \theta P}, 
    \end{equation}
    for a rotation angle $\theta$ and Pauli string $P$. Given a Pauli string $Q$, the evolved operator is 
    \begin{equation}
        V_k^{\dagger } Q V_k = 
        \begin{cases}
            Q, & \commu{P}{Q} = 0 \\
            \cos(\theta) Q - i \sin(\theta) P Q,  &\{P, Q\}=0.
        \end{cases}
    \end{equation}
    Propagating a single Pauli string with a Pauli gate results in at most $2$ Paulis, and across the layer produces at most $2^{K_l}$ output terms. 
    Therefore, the propagation time for symmetry-merged propagation is
    \begin{equation}
        T_{p,l} = \mathcal{O}(2^{K_l}\abs{\tilde{O}_l}),
    \end{equation}
    compared to $\mathcal{O}(2^{K_l}\abs{O_l})$ for standard Pauli propagation without merging.

    Combining the two contributions,
    \begin{equation}
        T_{\mathrm{sym}} = \sum_{l=1}^{L} \Big[ \mathcal{O}(\abs{\tilde{O}_{l-1}}\abs{G}) + \mathcal{O}(2^{K_l}\abs{\tilde{O}_l}) \Big].
    \end{equation}
    Define an effective reduction ratio $r_l = \abs{\tilde{O}_l} / \abs{O_l}$, which is upper bounded by the ratio in Theorem~\ref{thm:space-complexity} as $r_l < r$ , the total runtime relative to standard Pauli propagation scales as
    \begin{equation}
        \frac{T_{\mathrm{sym}}}{T_{\mathrm{PP}}} 
        \simeq 
        \frac{\sum_l (2^{K_l}r_l\abs{O_l} + \abs{G} \abs{U_l^{\dagger}\widetilde{O}_{l} U_l}) )}
             {\sum_l 2^{K_l}\abs{O_l}}.
    \end{equation}

    The symmetry Pauli propagation thus improves the time complexity of Pauli propagation by reducing the number of Pauli trajectories that must be evolved, at the cost of a modest merging overhead that scales with the symmetry group size.
    
\end{proof}
    We have formalized the time complexity of the symmetry Pauli propagation and provide a summary for examples of groups considered in this work in \tableref{tab:symmetry-complexity}. In practice, the overhead from merging Paulis is often subdominant, as the merging step can be implemented using canonical-representative computations or lookup tables.
    We illustrate the practical time saving in \supfigref{supfig1:ising_time}.
\newpage

\newpage

\begin{table}[h!]
\centering

\label{tab:symmetry-complexity}
\vspace{5mm}
\begin{tabular}{|l|ccc|}
\toprule
\textbf{Symmetry group $G$} & \textbf{Group order $\abs{G}$} 
& \textbf{Space complexity }  & \textbf{Reduction factor $r$ (asymptotic)}\\ 
\hline
No symmetry ($G=\{e\}$) & $1$ & $4^n$  & $1$ \\

Translation $\mathbb{Z}_n$ & $n$  & $\frac{1}{n}(4^n + 4(n-1)))$ & $\sim 1/n$ \\

Dihedral $D_n$ & $2n$  & $\cdot$ & $\sim 1/(2n)$ \\

Permutation $S_n$ & $n!$ & $\binom{n+3}{3} $ & $n^3 / 4^n$\\

General finite group $G$ & $\abs{G}$  & $\abs{\mathcal{R}^G_n}$ & $\sim \abs{\mathcal{R}^G_n}/4^n$\\
\hline
\end{tabular}
\caption{\textbf{Summary of computational complexity under different symmetry groups.}
Here $\abs{G}$ is the group order, $n$ the number of qubits, and $K_l$ the number of Pauli gates per layer.
The reduction factor $r$ defined in \thmref{thm:space-complexity} quantifies the fraction of orbit representatives retained after symmetry merging, assuming an observable with the support of the full set of Pauli strings.
}
\vspace{2mm}
\end{table}

\begin{figure}[h!]
    \centering
    \includegraphics[width=0.4\linewidth]{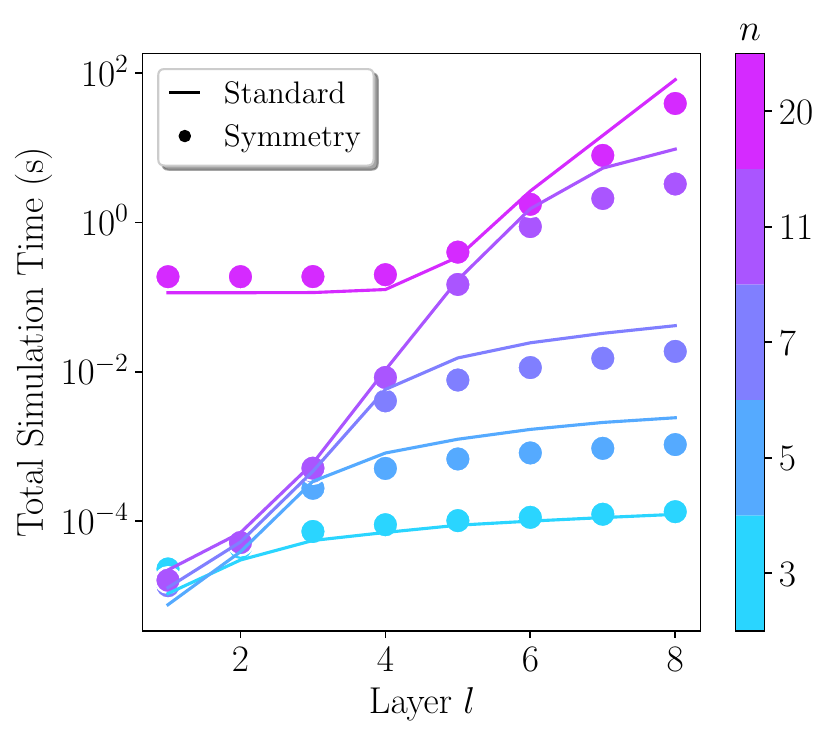}
    \caption{\textit{Simulation time in $1d$ Ising dynamics.} Scaling of the computational time ($second$) with circuit layers $l$ for system sizes $n = 3$ to $n=20$ (color bar). Circles denote standard PP and triangles denote the symmetry PP under translation symmetry.  }
    \label{supfig1:ising_time}
\end{figure}

\newpage
\section{Symmetry-merging Subroutines for Pauli Propagation}\label{app:merging-subroutines}

In this appendix, we outline the practical steps for the merging subroutine in step~\hyperref[alg:merging]{2} that applies to any symmetry group. This process identifies all Pauli strings that are equivalent under symmetry transformations and merges them into a single representative from their orbit. By reducing redundant terms, we improve the efficiency of Pauli propagation. First we outline the exhaustive merging procedure in \appref{app:merging_exhaustive}. Finally we outline the practical algorithms via lowest integer mapping for symmetry merging subroutines in \appref{app:merging_subroutines}.

\subsection{A general subroutine: exhaustive search}\label{app:merging_exhaustive}

When propagating an observable $O$ through a quantum circuit, the number of Pauli strings can grow exponentially. However, if the system has a symmetry group $G$, many of these Pauli strings are redundant because they transform into one another under symmetry operations. Instead of tracking all Pauli terms individually, we can group them into symmetry orbits and only propagate a reduced set of representatives.

At the $l$-th layer of the circuit, the observable is rewritten as:
\begin{equation}
    \widetilde{O}_l = \sum_{s \in \mathcal{R}^{G}_n}\sum_{g \in G} \Tr[O_l\, g \cdot (s)]s,
\end{equation}
where $\mathcal{R}^{G}_n$ are the non-equivalent orbit representatives (defined in \defref{def:orbit-repr}). After merging, only this reduced set of Pauli strings needs to be propagated through the circuit. The merging procedure follows these steps:

\begin{enumerate}
    \item \textbf{Initialize the merged Pauli Set $\widetilde{\mathcal{D}} = \varnothing$:}  
Start with a set $\mathcal{D}$ of Pauli strings with nonzero amplitudes in the decomposition of the observable $O$:
        \begin{equation}
            \mathcal{D} = \{(s, \Tr[O s]) \mid \Tr[O s] \neq 0, \, s \in \mathcal{P}_n\}.
        \end{equation}

    \item \textbf{Loop Over Each Pauli String and Check Symmetry Equivalence:}  
    For each Pauli string $s$ in $\mathcal{D}$, check whether it belongs to an existing symmetry orbit. Apply all symmetry transformations $g \in G$ to generate the orbit of $s$:
        \begin{equation}
            \operatorname{Orb}_G(s) = \{ g \cdot s \mid g \in G \}.
        \end{equation}
        If another Pauli string $\widetilde{s}$ in $\widetilde{\mathcal{D}} = \{\left(\widetilde{s}, a_{\widetilde{s}}\right)\}$ is already in the orbit, so that $\widetilde{s} = g \cdot s$ for some $g$, merge their coefficients:
        \begin{equation}
            a_{\widetilde{s}} \leftarrow a_{\widetilde{s}} + \Tr[O s].
        \end{equation}
        Otherwise, store $s$ as a new representative in $\widetilde{\mathcal{D}} \leftarrow \widetilde{\mathcal{D}} + \left(s,  \Tr[O s] \right) $.

    \item \textbf{Construct the Merged Observable:}  
    The final grouped observable is given by:
        \begin{equation}
            \widetilde{O} = \sum_{s \in \widetilde{\mathcal{D}}} a_s s.
        \end{equation}
    
\end{enumerate}

While this exhaustive merging procedure significantly reduces memory overhead by tracking only symmetry-distinct Pauli strings, it can be computationally expensive when the symmetry group $G$ is large. To address this, alternative merging methods, such as canonical forms given by the lowest integer representations (see Algorithm~\hyperref[alg:translation-int]{2} and Algorithm~\hyperref[alg:permutation-int]{3}), may be more practical especially for certain symmetries like permutation groups.

\subsection{Practical subroutines: lowest integer mapping}\label{app:merging_subroutines}

In this section, we discuss our practical implementations considering two examples of translation and permutation symmetry. In general, finding a minimal set of representatives is a challenging task for an arbitrary symmetry. 
In this work we adopt a canonical mapping based on the \emph{lowest integer representation}, which provides a consistent and symmetry-agnostic way to identify orbit representatives.

\begin{definition}[\textbf{Lowest integer representation}]\label{def:low-int}
    Let $G$ be the symmetry group, and $\operatorname{Orb}_G(s_0)$ be the orbits of a Pauli string $s_0$. Define a binary representation of a Pauli string by the a base 4 numeral such that a pair of bits represent a Pauli operator 
    \begin{equation}
        I \rightarrow 00, \,\,\,
        X \rightarrow 01,\,\,\,
        Y \rightarrow 10,\,\,\,
        Z \rightarrow 11.
    \end{equation}
    For an $n$-qubit Pauli string, this defines a binary vector $b \in \mathbb{Z}_2^{2n}$. 
        Its corresponding integer representation is
        \begin{equation}
            \operatorname{int}: \mathbb{Z}_2^{2n} \rightarrow \mathbb{Z}, 
            \qquad 
            \operatorname{int}(b) = \sum_{i=0}^{2n-1} b_i \, 2^i.
        \end{equation}
        The \emph{lowest integer representation} of $s_0$ is then defined as
        \begin{equation}
            \mathrm{argmin}_{s \in \operatorname{Orb}_G(s_0)} \operatorname{int}(s).
        \end{equation}
\end{definition}

The binary representation convention is equivalent to the more familiar \emph{symplectic representation} of Pauli strings,
\begin{equation}
    P = \prod_{i=1}^n X^{a_i} Z^{b_i}, \qquad a_i, b_i \in \{0, 1\},
\end{equation}
where multiplication of Pauli operators corresponds to addition modulo two in the symplectic vector space~\cite{rudolph2025pauli}. 
Here, we follow the base-$4$ numeral conventions as Ref.~\cite{rudolph2025pauli}, to which we refer the reader for further implementation details.

Finally, we outline two examples of symmetry merging using the lowest integer representations of Pauli strings for translational symmetry (Algorithm~\hyperref[alg:translation-int]{2}) and permutation symmetry (Algorithm~\hyperref[alg:permutation-int]{3}).

\noindent\rule{17.5cm}{0.8pt}

\noindent
\textbf{Algorithm 2} Symmetry merging subroutine for translation $\mathbb{Z}_n$ (Lowest Integer Representation)\label{alg:translation-int}
\vspace{-2mm}

\noindent\rule{17.5cm}{0.4pt}

\begin{algorithmic}[1]\label{alg:zn-low-int}
\Require Observable $O = \sum_{s \in \mathcal{P}_n} \Tr[Os]\, s$. \\
Set of Pauli–coefficient pairs $\mathcal{D} = \{ (\Tr[Os], s)\, |\, \Tr[Os] \neq 0 \}$.
\Ensure Observable $\tilde{O}$ after merging Pauli strings within the same symmetry orbit.
\vspace{2mm}

\State {\bf Initialize:} $\widetilde{\mathcal{D}} \gets \varnothing$. \Comment{Initialize an empty merged set.}
\For{$(\Tr[Os], s) \in \mathcal{D}$} \Comment{Enumerate all Pauli strings in $O$.}
    \State Compute the lowest integer representation:
    \For{$j = 0, \dots, n-1$}
        \State Compute $\operatorname{int}(\hat{T}^{j}(s))$.
    \EndFor
    \State Set $s_0 \gets \mathrm{argmin}_{s' \in \operatorname{Orb}_G(s)} \operatorname{int}(s')$. \Comment{Orbit representative.}
    \If{$(\cdot, s_0) \in \widetilde{\mathcal{D}}$}
        \State Update coefficient: $a_{s_0} \gets a_{s_0} + \Tr[Os]$.
    \Else
        \State Append pair $(\Tr[Os], s_0)$ to $\widetilde{\mathcal{D}}$.
    \EndIf
\EndFor
\State \textbf{Output:} Merged observable $\widetilde{O} = \sum_{(a_s, s) \in \widetilde{\mathcal{D}}} a_s s$.
\end{algorithmic}

\vspace{-2mm}

\noindent\rule{17.5cm}{0.8pt}

Notice that computing the lowest integer representation generally requires iterating over all elements within the same group orbit. 
For example, under translational symmetry $\mathbb{Z}_n$, this incurs a computational cost proportional to the group size, $\abs{\mathbb{Z}_n} = n$. 
However, this scaling with $\abs{G}$ is not always necessary. 
A notable exception occurs for permutation symmetry, as illustrated in Algorithm~\hyperref[alg:permutation-int]{3}, where the lowest integer representation can be determined directly without explicit enumeration of all group elements.
This is the case for permutation symmetry described next in Algorithm~\hyperref[alg:permutation-int]{3}. Additionally, finding this representative is embarrassingly parallel between the Pauli strings, which lends itself to parallelization on CPU or GPU.

\newpage
\noindent\rule{17.5cm}{0.8pt}

\noindent
\textbf{Algorithm 3} Symmetry merging subroutine for permutation $S_n$ (Lowest Integer Representation) \label{alg:permutation-int}
\vspace{-2mm}

\noindent\rule{17.5cm}{0.4pt}

\begin{algorithmic}[1]
\Require Observable $O = \sum_{s \in \mathcal{P}_n} \Tr[Os]\, s$. \\
Set of Pauli–coefficient pairs $\mathcal{D} = \{ (\Tr[Os], s)\, |\, \Tr[Os] \neq 0 \}$.
\Ensure Observable $\tilde{O}$ after merging Pauli strings within the same symmetry orbit.
\vspace{2mm}

\State {\bf Initialize:} $\widetilde{\mathcal{D}} \gets \varnothing$. \Comment{Initialize an empty merged set.}
\For{$(\Tr[Os], s) \in \mathcal{D}$} \Comment{Enumerate all Pauli strings in $O$.}
    \State Compute the lowest integer representation:
    \State Count the number of occurrences of each non-identity Pauli as $(n_X, n_Y, n_Z)$.
    Set $s_0 \gets \operatorname{int}(Z^{\otimes n_Z} Y^{\otimes n_Y} X^{\otimes n_X})$. \Comment{Orbit representative.} \label{lag:perm-represent}
    \If{$(\cdot, s_0) \in \widetilde{\mathcal{D}}$}
        \State Update coefficient: $a_{s_0} \gets a_{s_0} + \Tr[Os]$.
    \Else
        \State Append pair $(\Tr[Os], s_0)$ to $\widetilde{\mathcal{D}}$.
    \EndIf
\EndFor
\State \textbf{Output:} Merged observable $\widetilde{O} = \sum_{(a_s, s) \in \widetilde{\mathcal{D}}} a_s s$.
\end{algorithmic}

\vspace{-2mm}
\noindent\rule{17.5cm}{0.4pt}

In Line~\hyperref[lag:perm-represent]{5}, the lowest integer representative is determined simply by counting the number of each Pauli type $(X, Y, Z)$ within a string. The direct counting is significantly more efficient than enumerating over the entire permutation group, which scales as $\abs{S_n}=n!$.

\end{appendix}

\end{document}